\documentclass[11pt,a4paper]{article}
\usepackage[a4paper,margin=2.9cm]{geometry}
\usepackage{amsmath,amssymb,amsthm,mathtools}
\usepackage{graphicx}
\usepackage{xcolor}
\usepackage{microtype}
\usepackage{comment}
\usepackage[hidelinks]{hyperref}
\usepackage{authblk}
\hypersetup{colorlinks=true, citecolor=magenta, linkcolor=magenta, urlcolor=magenta}
\usepackage[
  style=alphabetic,
  sorting=nyt,
  minalphanames=3,
  maxalphanames=4,
  maxbibnames=99,
  doi=false,
  url=false,
  isbn=false,
  eprint=false,
  backend=biber
]{biblatex}
\DeclareFieldFormat{linked}{%
  \iffieldundef{doi}
    {\iffieldundef{url}{#1}{\href{\thefield{url}}{#1}}}
    {\href{https://doi.org/\thefield{doi}}{#1}}}
\renewbibmacro{in:}{}
\DeclareFieldFormat{pages}{#1}
\DeclareFieldFormat[article]{volume}{\mkbibbold{#1}}
\DeclareFieldFormat[article]{number}{(#1)}
\renewbibmacro*{volume+number+eid}{%
  \printfield{volume}%
  \printfield{number}}
\renewbibmacro*{journal+issuetitle}{%
  \printtext[linked]{%
    \usebibmacro{journal}%
    \setunit*{\addspace}%
    \usebibmacro{volume+number+eid}%
    \setunit{\bibpagespunct}%
    \printfield{pages}%
    \setunit*{\addspace}%
    \printtext[parens]{\printdate}}}
\renewbibmacro*{note+pages}{\printfield{note}\newunit}
\renewbibmacro*{publisher+location+date}{%
  \printtext[linked]{%
    \printlist{publisher}%
    \setunit{\addcomma\space}%
    \printdate}%
  \newunit}
\theoremstyle{plain}
\newtheorem{theorem}{Theorem}
\newtheorem{lemma}{Lemma}

\newtheorem{corollary}{Corollary}
\theoremstyle{definition}
\newtheorem{definition}{Definition}

\theoremstyle{remark}
\newtheorem{remark}{Remark}
\theoremstyle{definition}

\definecolor{amethyst}{rgb}{0.75, 0, 1}

\title{\textbf{Monotonicity of the R\'enyi channel capacity\\
       under non-signaling assisted channel simulation}}
       
\author[1]{Tina Shariat\thanks{\href{mailto:tina.shariatt@gmail.com}{tina.shariatt@gmail.com}}}
\author[1,2]{Matt Hoogsteder-Riera\thanks{\href{mailto:hoogstedermatt@gmail.com}{hoogstedermatt@gmail.com}}}
\author[1,3]{Marco Tomamichel\thanks{\href{mailto:marco.tomamichel@nus.edu.sg}{marco.tomamichel@nus.edu.sg}}}

\affil[1]{\textit{Centre for Quantum Technologies,\protect\\[-1mm] National University of Singapore, Singapore 117543}}
\affil[2]{\textit{ Campus for Research Excellence and Technological Enterprise (CREATE),\protect\\[-1mm] National University of Singapore,\protect\\[-1mm] 1 Create Way, Singapore 13860}}
\affil[3]{\textit{Department of Electrical and Computer Engineering,\protect\\[-1mm] National University of Singapore, Singapore 117583}}

\date{\vspace*{-10mm}}
\begin{document}
\maketitle
\begin{abstract}
We give an alternative variational proof that non-signaling correlations shared between the sender and receiver cannot increase the R\'enyi capacity of a classical channel. This monotonicity follows from the data-processing inequality for channel divergences established by Oufkir, Yao and Berta~\cite[Lemma~24]{Oufkir2026Strong}. Our approach works directly with the information-radius representation of the capacity. The constraints defining non-signaling boxes are linear in the box, while the R\'enyi capacity depends on the channel nonlinearly. We use a Legendre transform to express the capacity as an extremum over a family of affine functionals in which the channel appears linearly. Combining this method with separate arguments for the Shannon capacity and endpoint orders, we recover monotonicity at every order from zero to infinity. We then derive consequences for standard coding exponent functions: non-signaling simulation cannot increase the random-coding or sphere-packing exponent below capacity, or decrease the strong-converse exponent above capacity.
\end{abstract}
\section{Introduction}
\label{sec:intro}

The capacity of a classical channel gives the largest rate at which reliable communication is possible. It does not say how reliable a code of a given rate can be. Below capacity, one wants to know how fast the error probability can go to zero with the block length; above capacity, one wants to know how fast the success probability must go to zero. Standard bounds for both questions are expressed through the one-parameter family of R\'enyi capacities $R_\alpha$, indexed by the order $\alpha$. At $\alpha = 1$ the family gives the Shannon capacity. Orders $\alpha \in (0,1)$ govern the error exponents below capacity: the random-coding exponent of Gallager~\cite{gallager1965simple} uses $\alpha \in [\tfrac12, 1)$, and the sphere-packing bound of Shannon, Gallager and Berlekamp~\cite{shannon1967lower} uses $\alpha \in (0,1)$. Orders $\alpha > 1$ govern the strong converse exponent of Arimoto~\cite{arimoto1973converse}, which Dueck and K\"orner~\cite{DueckKorner1979} showed to be tight. Therefore, to describe the reliability of a channel at every rate, and not only its capacity, one needs the whole family of R\'enyi capacities.

The R\'enyi capacities come from the divergence measures introduced by R\'enyi as a one-parameter generalization of relative entropy~\cite{renyi1961}. They have a geometric meaning: Csisz\'ar showed that the $\alpha$-capacity equals the radius of the channel's set of output distributions, measured with the R\'enyi divergence~\cite{csiszar1995}; the same statement for the Shannon capacity goes back to Kemperman~\cite{kemperman1974}, and the structure of the optimal center was studied in detail by Nakibo\u{g}lu~\cite{nakiboglu2019renyi,nakiboglu2019augustin}. Our proofs work directly with this radius formulation rather than with the maximization over input distributions.

In practice, a channel is rarely used alone. A real protocol adds processing before the channel at the sender and after the channel at the receiver, and the result is again a channel. The general model for this composition is a \emph{super-channel}: a processing box that takes an outer input, prepares an input for the base channel, and then maps the channel output to an outer output. The key feature of this model is that the input side and the output side of the box need not act independently. They may share resources: shared randomness, entanglement, or, most generally, any correlation that causality allows. Correlations of the last type are called \emph{non-signaling boxes}. They were introduced by Popescu and Rohrlich to study which physical principles separate quantum mechanics from more general theories~\cite{popescu1994nonlocality}, and they form a convex set strictly larger than the set of quantum correlations.

Because this class of correlations is so large, it is a natural relaxation of coding problems, and the relaxation has been very productive. Matthews showed that optimal coding with non-signaling assistance is a linear program, and that its value is exactly the finite-blocklength converse of Polyanskiy, Poor and Verd\'u~\cite{matthews2012linear,polyanskiy2010channel}: a bound first obtained by other methods is the exact answer to an assisted coding task. Cubitt, Leung, Matthews and Winter used non-local correlations to study zero-error capacities and channel simulation: with assistance, these quantities have simple formulas, while the unassisted versions remain intractable~\cite{cubitt2011zero}. Non-signaling assistance is therefore a useful tool for making otherwise intractable operational quantities computable.

In this work, we give an alternative proof that the R\'enyi capacity of a classical channel is monotone under non-signaling simulation. Two points motivate studying this property. First, a general non-signaling box need not decompose into separate encoder and decoder maps, so the data processing arguments that apply in the separable regime cannot be used directly. Second, an answer at $\alpha = 1$ alone would not settle the other orders. Monotonicity of the Shannon capacity does not exclude a box that keeps the capacity fixed but improves the error exponent curves below capacity, or weakens the strong converse above it. We show that neither is possible: the entire R\'enyi capacity family is monotone under non-signaling simulation.

This monotonicity statement follows from the data-processing inequality for channel divergences established by Oufkir, Yao and Berta~\cite[Lemma~24]{Oufkir2026Strong}. Indeed, applying their inequality to a channel and a constant-output channel, and minimizing over the constant output distribution, gives the R\'enyi-capacity inequality; the endpoint orders follow by limits. We became aware of this implication after completing our proof. Our contribution is an alternative derivation based on a variational representation of the information radius.

Our proof uses a variational representation to express the nonlinear R\'enyi-capacity comparison in a form compatible with the linear non-signaling constraints. Our main tool is a Legendre transform: using Young's inequality, we rewrite the capacity as an extremum over a family of affine functionals, and in this formulation, the channel components appear in linear form. The sign of the H\"older conjugate of $\alpha$ then decides whether the extremum is a supremum attained by an explicit function or an infimum bounded uniformly, and this separates the regimes of the proof.

The separate treatment of the endpoint orders also clarifies which non-signaling conditions are needed. In particular, the proof at $\alpha=\infty$ uses only the no-forward-signaling condition, so the corresponding monotonicity statement holds for a larger class of processing boxes, provided that the simulation formula still defines a valid channel, as noted after Lemma~\ref{lem:alphainf}.

Beyond proving the theorem, we explain what it means for coding. Using Csisz\'ar's identity between the R\'enyi capacities and Gallager's function~\cite{csiszar1995}, the monotonicity of the family becomes a statement about exponents: the simulated channel never has a larger random-coding~\cite{gallager1965simple} or sphere-packing~\cite{shannon1967lower} exponent than the base channel, and above capacity, its success probability never decays more slowly (Corollary~\ref{cor:curves}). This is true for every rate independently, so a given box cannot improve the exponent at one rate by accepting a worse one at another; and wherever the reliability function is known --- above the critical rate --- the guarantee applies to the reliability function itself, not only to its bounds (see Section~\ref{sec:disc-open}). The second consequence is about assisted communication directly. We view the encoder, the $n$ channel uses, and the decoder of an assisted code as a single box that simulates a channel from the message to the decoded message. The theorem then gives an Arimoto-type bound~\cite{arimoto1973converse} that holds for every non-signaling assisted code at every block length: above capacity, the success probability decays exponentially, with the same exponent that already governs unassisted codes~\cite{DueckKorner1979} (Corollary~\ref{prop:sc}; see also Corollary~\ref{cor:alphainf-code}). Non-signaling assistance, including the physically relevant shared randomness and entanglement, cannot weaken the strong converse exponent of a discrete memoryless channel. The equality of the assisted and unassisted strong converse exponents is also covered by the channel-interconversion results of~\cite[Theorem~3 and Remark~1]{Oufkir2026Strong}. Here we recover this conclusion through R\'enyi-capacity monotonicity.

The remainder of this paper is structured as follows. Section~\ref{sec:prelim} fixes the notation and reviews the definitions of R\'enyi divergence, the information radius, the super-channel framework, and the non-signaling conditions. Section~\ref{sec:monotonicity} proves the monotonicity results in the five regimes of $\alpha$ in Subsections~\ref{sec:alphaeq1} through~\ref{sec:disc-zero}, and Subsection~\ref{sec7main} combines them into the main theorem. Section~\ref{sec:applications} develops consequences of the monotonicity for channel simulation and coding. Finally, Section~\ref{sec:discussion} contrasts them with the gains that assisted coding achieves and closes with some open questions.
\section{Preliminaries}
\label{sec:prelim}

\subsection{Information measures}
\label{subsec:info-measures}

We use the standard notation for Shannon entropy $H(X)$, conditional entropy $H(X\mid Y)$, and mutual information $I(X;Y)$. The conditional mutual information is defined as:
\begin{equation}
  I(X;Y\mid Z) = H(X\mid Z) - H(X\mid Y,Z).
  \label{eq:cmi}
\end{equation}
We employ the chain rule in its two standard forms,
\begin{align}
  I(X;Y,Z) &= I(X;Y) + I(X;Z\mid Y) \label{eq:chain-rule-first} \\
             &= I(X;Z) + I(X;Y\mid Z), \label{eq:chain-rule-second}
\end{align}
along with the non-negativity of conditional mutual information,
\begin{equation}
  I(X;Y\mid Z) \ge 0,
  \label{eq:cmi-nonnegative}
\end{equation}
where equality holds if and only if $X$ and $Y$ are conditionally independent given $Z$ \cite{cover2006elements}.

\subsection{The Shannon capacity}
\label{subsec:shannon-capacity}

Let $X$ and $Y$ be random variables taking values in the input alphabet $\mathcal{X}$ and the output alphabet $\mathcal{Y}$, respectively. The Shannon capacity of a discrete memoryless channel $W(y\mid x)$ is defined as:
\begin{equation}
  C(W) \coloneqq \max_{P_X \in \mathcal{P}(\mathcal{X})} I(X;Y),
  \label{eq:shannon-capacity}
\end{equation}
where the joint distribution of $(X,Y)$ is given by $P_X(x) W(y\mid x)$ \cite{shannon1948}, and $\mathcal{P}(\mathcal{X})$ denotes the simplex of probability distributions over $\mathcal{X}$. The maximum in \eqref{eq:shannon-capacity} is attained because the mapping $P_X \mapsto I(X;Y)$ is concave and continuous on the compact set $\mathcal{P}(\mathcal{X})$.

\subsection{Super-channels and non-signaling conditions}
\label{subsec:super-channel}

A non-signaling assisted simulation of a channel $W'$ from a base channel $W$ is mediated by a processing box $\Pi(x,y'\mid x',y)$ wrapped around $W$, as illustrated in Figure~\ref{fig:super-channel}. The box generates the input of the base channel $x$ and the simulator output $y'$ based jointly on $x'$ and $y$. To prevent any information transfer bypassing $W$, the box must satisfy the non-signaling conditions of Definition~\ref{def:ns-simulation}: the choice of $x$ cannot anticipate $y$, and the generation of $y'$ cannot directly access $x'$.

\begin{definition}[Non-signaling simulation]
\label{def:ns-simulation}
A channel $W'$ is said to be \emph{non-signaling simulable} from a channel $W$ if there exists a joint distribution box $\Pi$ such that for all inputs and outputs:
\begin{align}
  W'(y'\mid x')
    &= \sum_{x,y} \Pi(x,y'\mid x',y)\,W(y\mid x),
    \label{eq:ns-simulation} \\
  E(x\mid x')
    &\coloneqq \sum_{y'} \Pi(x,y'\mid x',y)
      \quad \text{is independent of } y,
    \label{eq:ns1} \\
  F(y'\mid y)
    &\coloneqq \sum_{x} \Pi(x,y'\mid x',y)
      \quad \text{is independent of } x' .
    \label{eq:ns2}
\end{align}
We refer to \eqref{eq:ns1} as the \emph{no-backward signaling} condition and to \eqref{eq:ns2} as the \emph{no-forward signaling} condition.
\end{definition}

\begin{figure}[htbp]
  \centering
  \includegraphics[scale=1]{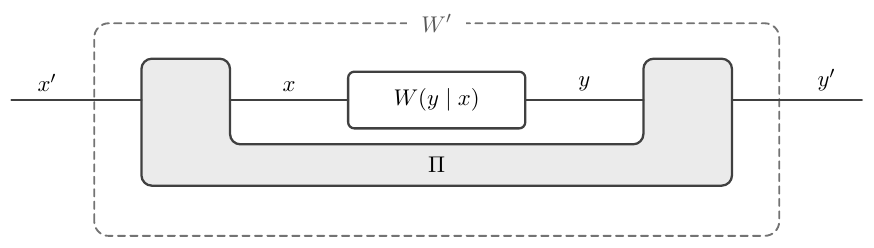}
  \caption{Non-signaling simulation of $W'$ from $W$ (Definition~\ref{def:ns-simulation}).}
  \label{fig:super-channel}
\end{figure}

Summing \eqref{eq:ns1} over $x$ and \eqref{eq:ns2} over $y'$ confirms that $E(\cdot\mid x')$ and $F(\cdot\mid y)$ form valid probability distributions on $\mathcal{X}$ and $\mathcal{Y}'$, respectively. Furthermore, the non-signaling constraints guarantee that $W'$ itself constitutes a valid conditional probability distribution (channel), as verified by the following derivation:
\begin{align}
  \sum_{y'} W'(y'\mid x')
    &= \sum_{y'} \sum_{x,y} \Pi(x,y'\mid x',y)\,W(y\mid x) \label{eq:norm-1} \\
    &= \sum_{x,y} W(y\mid x) \sum_{y'} \Pi(x,y'\mid x',y) \label{eq:norm-2} \\
    &= \sum_{x,y} W(y\mid x)\,E(x\mid x') \label{eq:norm-3} \\
    &= \sum_{x} E(x\mid x') \sum_{y} W(y\mid x) \label{eq:norm-4} \\
    &= \sum_{x} E(x\mid x') \label{eq:norm-5} \\
    &= 1. \label{eq:norm-6}
\end{align}

\subsection{R\'enyi divergence and the information radius}
\label{subsec:renyi}
Let $P$ and $Q$ be probability distributions over a finite alphabet $\mathcal{X}$. For every order $\alpha \in (0,1) \cup (1,\infty)$, the R\'enyi divergence of order $\alpha$ is defined as:
\begin{equation}
  D_\alpha(P \,\|\, Q) \coloneqq \frac{1}{\alpha-1} \log \sum_{x} P(x)^{\alpha} Q(x)^{1-\alpha}.
  \label{eq:renyi-divergence}
\end{equation}
For $0<\alpha<1$, terms with $P(x)Q(x)=0$ vanish, and $D_\alpha(P\|Q)=+\infty$ if the sum is zero. For $\alpha>1$, a term with $P(x)>0=Q(x)$ equals $+\infty$, while terms with $P(x)=0$ vanish~\cite{renyi1961}. The orders $\alpha=0$, $\alpha = 1$ and $\alpha = \infty$ are defined by their corresponding limits, and in particular, $D_1(P \,\|\, Q)$ recovers the relative entropy \cite{vanerven2014}. The quantity we track throughout is the $\alpha$-information radius of a channel $W(y\mid x)$, defined as:
\begin{equation}
    R_\alpha(W) \coloneqq \min_{Q \in \mathcal{P}(\mathcal{Y})}\max_{x \in \mathcal{X}} D_\alpha\big(W(\cdot\mid x) \,\big\|\, Q\big),
  \label{eq:info-radius}
\end{equation}
where the minimizing distribution $Q$ ranges over the output simplex and the maximum is taken over the input alphabet. The radius \eqref{eq:info-radius} coincides with the R\'enyi channel capacity of order $\alpha$, and at $\alpha = 1$ it reduces to the Shannon capacity \eqref{eq:shannon-capacity}, so that $R_1(W) = C(W)$ \cite{kemperman1974,csiszar1995}. Consequently, the monotonicity properties established in Sections~\ref{sec:alphalt1}--\ref{sec:disc-zero} directly govern the behavior of the R\'enyi capacity via the minimax representation \eqref{eq:info-radius}.

\section{Monotonicity of the R\'enyi capacity}
\label{sec:monotonicity}
The proof of monotonicity splits\textit{} into five distinct regimes according to the order $\alpha$. On the open intervals $0 < \alpha < 1$ and $\alpha > 1$, the argument is variational, where the sign change of $\alpha-1$ reverses the inequalities accordingly. The critical case $\alpha = 1$ is fundamentally different: the R\'enyi divergence reduces to the relative entropy, the exponential weights vanish, and the monotonicity follows directly from classical information inequalities via the chain rule. At $\alpha = \infty$, the limit of the divergence admits a closed-form expression for the information radius, reducing the claim to a combinatorial inequality. Finally, $\alpha=0$ is resolved as a limit case of $0<\alpha<1$. These five regimes are established in the subsections below, paving the way for the main theorem consolidated in Subsection~\ref{sec7main}.

\subsection{Monotonicity for \texorpdfstring{$\alpha=1$}{alpha=1}}
\label{sec:alphaeq1}

\begin{lemma}
\label{lem:alphaeq1}
If the simulator channel $W'$ is non-signaling simulable from the base channel $W$, then their respective Shannon capacities satisfy the monotonicity relation:
\begin{equation}
  C(W') \le C(W).
  \label{eq:lemma-alphaeq1}
\end{equation}
\end{lemma}

\begin{proof}
Fix an arbitrary input distribution $P_{X'}\in\mathcal{P}(\mathcal{X}')$. We will prove that
\begin{equation}
    I(X';Y') \le I(X;Y),
    \label{eq:a1-mi-drop}
\end{equation}
where $X$ carries the input distribution that $P_{X'}$ induces on $W$ through the box. Throughout, logarithms are natural, and every sum over $(x',x,y,y')$ runs over the tuples with $p(x',x,y,y')>0$.

The random variables $X',X,Y,Y'$ have the joint distribution
\begin{equation}
    p(x',x,y,y') = P_{X'}(x')\,\Pi(x,y'\mid x',y)\,W(y\mid x).
    \label{eq:a1-joint}
\end{equation}
Summing over $y'$ and using the no-backward-signaling condition \eqref{eq:ns1},
\begin{equation}
    p(x',x,y) = P_{X'}(x')\,E(x\mid x')\,W(y\mid x),
    \label{eq:a1-joint3}
\end{equation}
so that
\begin{equation}
    p(y\mid x,x') = W(y\mid x),
    \qquad
    I(X';Y\mid X) = 0.
    \label{eq:a1-markov}
\end{equation}
Summing \eqref{eq:a1-joint} over $x$ and $y$ and using \eqref{eq:ns-simulation},
\begin{equation}
    p(y'\mid x') = W'(y'\mid x').
    \label{eq:a1-outer}
\end{equation}

By the chain rule,
\begin{align}
    I(X';Y) + I(X';Y'\mid Y) &= I(X';Y') + I(X';Y\mid Y'), \label{eq:a1-chain1}\\
    I(X';Y) + I(X;Y\mid X') &= I(X;Y) + I(X';Y\mid X). \label{eq:a1-chain2}
\end{align}
Subtracting \eqref{eq:a1-chain2} from \eqref{eq:a1-chain1} and using \eqref{eq:a1-markov},
\begin{equation}
    I(X;Y) - I(X';Y') = I(X;Y\mid X') + I(X';Y\mid Y') - I(X';Y'\mid Y).
    \label{eq:a1-identity}
\end{equation}
Writing $I(X';Y\mid Y') = H(X'\mid Y') - H(X'\mid Y,Y')$ and $I(X';Y'\mid Y) = H(X'\mid Y) - H(X'\mid Y,Y')$, the term $H(X'\mid Y,Y')$ cancels, and we obtain
\begin{equation}
    I(X;Y) - I(X';Y') = I(X;Y\mid X') + H(X'\mid Y') - H(X'\mid Y).
    \label{eq:a1-identity-H}
\end{equation}
Hence \eqref{eq:a1-mi-drop} is equivalent to
\begin{equation}
    I(X;Y\mid X') + H(X'\mid Y') - H(X'\mid Y) \ge 0.
    \label{eq:a1-entropy}
\end{equation}

We now write every term in \eqref{eq:a1-entropy} with the joint distribution \eqref{eq:a1-joint}:
\begin{align}
    I(X;Y\mid X') &= \sum p(x',x,y,y')\,\log\frac{p(y\mid x,x')}{p(y\mid x')}, \label{eq:a1-I}\\
    H(X'\mid Y') &= -\sum p(x',x,y,y')\,\log p(x'\mid y'), \label{eq:a1-H1}\\
    H(X'\mid Y) &= -\sum p(x',x,y,y')\,\log p(x'\mid y). \label{eq:a1-H2}
\end{align}
Therefore, the left-hand side of \eqref{eq:a1-entropy} equals
\begin{equation}
    \sum p(x',x,y,y')\,\log A,
    \label{eq:a1-sum-logA}
\end{equation}
where
\begin{equation}
    A \coloneqq \frac{p(y\mid x,x')}{p(y\mid x')}\cdot\frac{p(x'\mid y)}{p(x'\mid y')}.
    \label{eq:a1-A-def}
\end{equation}
By Bayes' rule and \eqref{eq:a1-outer},
\begin{align}
    p(x'\mid y) &= \frac{P_{X'}(x')\,p(y\mid x')}{p(y)}, \label{eq:a1-bayes1}\\
    p(x'\mid y') &= \frac{P_{X'}(x')\,W'(y'\mid x')}{p(y')}. \label{eq:a1-bayes2}
\end{align}
Substituting \eqref{eq:a1-bayes1}, \eqref{eq:a1-bayes2}, and \eqref{eq:a1-markov} into \eqref{eq:a1-A-def}, the factors $P_{X'}(x')$ and $p(y\mid x')$ cancel, and $A$ simplifies to
\begin{equation}
    A = \frac{W(y\mid x)\,p(y')}{p(y)\,W'(y'\mid x')}.
    \label{eq:a1-A}
\end{equation}

Since $\log t \le t-1$ for $t>0$ by concavity of the logarithm, taking $t=1/A$ gives $\log A \ge 1-1/A$. Multiplying by $p(x',x,y,y')$ and summing, we obtain
\begin{equation}
    \sum p(x',x,y,y')\,\log A \ge 1 - \sum p(x',x,y,y')\,\frac{1}{A},
    \label{eq:a1-reduced}
\end{equation}
and it remains to show that the last sum is at most one. By \eqref{eq:a1-joint} and \eqref{eq:a1-A}, the factor $W(y\mid x)$ cancels:
\begin{equation}
    p(x',x,y,y')\,\frac{1}{A} = P_{X'}(x')\,\Pi(x,y'\mid x',y)\,\frac{p(y)\,W'(y'\mid x')}{p(y')}.
    \label{eq:a1-cancel}
\end{equation}
 Let $\mathcal{S}\coloneqq\{(x',y') : P_{X'}(x')\,W'(y'\mid x')>0\}$. Every tuple with $p(x',x,y,y')>0$ satisfies $(x',y')\in\mathcal{S}$, and $p(y')>0$ for every $(x',y')\in\mathcal{S}$. The right-hand side of \eqref{eq:a1-cancel} is nonnegative, so extending the sum over $(x,y)$ to all pairs can only increase it. We can act on the average of $A^{-1}$ to find the result:
\begin{align}
    \sum p(x',x,y,y')\,\frac{1}{A}
    &\le \sum_{(x',y')\in\mathcal{S}}\frac{P_{X'}(x')\,W'(y'\mid x')}{p(y')}\sum_{x,y}\Pi(x,y'\mid x',y)\,p(y) \label{eq:a1-sum1}\\
    &= \sum_{(x',y')\in\mathcal{S}}\frac{P_{X'}(x')\,W'(y'\mid x')}{p(y')}\sum_{y}F(y'\mid y)\,p(y) \label{eq:a1-sum2}\\
    &= \sum_{y':\,p(y')>0}\frac{\sum_{y}F(y'\mid y)\,p(y)}{p(y')}\underbrace{\sum_{x':\,(x',y')\in\mathcal{S}}P_{X'}(x')\,W'(y'\mid x')}_{=\,p(y')} \label{eq:a1-sum3}\\
    &= \sum_{y':\,p(y')>0}\;\sum_{y}F(y'\mid y)\,p(y) \label{eq:a1-sum4}\\
    &\le \sum_{y'}\sum_{y}F(y'\mid y)\,p(y) = \sum_{y}p(y)\sum_{y'}F(y'\mid y) = 1. \label{eq:a1-sum5}
\end{align}
Here, \eqref{eq:a1-sum2} follows from the no-forward-signaling condition \eqref{eq:ns2}; \eqref{eq:a1-sum3} regroups the sum by $y'$, which is possible because $\sum_{y}F(y'\mid y)\,p(y)$ does not depend on $x'$; and the inequality in \eqref{eq:a1-sum5} holds because it only adds the nonnegative terms with $p(y')=0$. Combining \eqref{eq:a1-sum5} with \eqref{eq:a1-reduced}, the left-hand side of \eqref{eq:a1-entropy} is at least $1-1=0$, which proves \eqref{eq:a1-mi-drop}.

Finally, taking $P_{X'}$ capacity-achieving for $W'$,
\begin{equation}
    C(W') = I(X';Y') \le I(X;Y) \le C(W),
    \label{eq:a1-final}
\end{equation}
where the last inequality holds because the induced input distribution of $X$ need not be optimal for $W$.
\end{proof}

\medskip

\subsection{Monotonicity for \texorpdfstring{$0<\alpha<1$}{0<alpha<1}}
\label{sec:alphalt1}

\begin{lemma}
\label{lem:alphalt1}
If the simulator channel $W'$ is non-signaling simulable from the base channel $W$, then their respective R\'enyi capacities of order $0 < \alpha < 1$ satisfy the monotonicity relation:
\begin{equation}
  R_\alpha(W') \le R_\alpha(W).
  \label{eq:lemma-alphalt1}
\end{equation}
\end{lemma}

\begin{proof}
Throughout the proof, $\alpha \in (0,1)$ is held fixed. To eliminate logarithms and simplify the analytical expressions, we introduce the exponentiated surrogate as follows:
\begin{equation}
  S_\alpha(W) \coloneqq e^{(\alpha-1) R_\alpha(W)} .
  \label{eq:surrogate-def-lt1}
\end{equation}
Since $\alpha < 1$, $e^{(\alpha-1)(\cdot)}$ is strictly decreasing, so comparing $R_\alpha$ is equivalent to comparing $S_\alpha$ in reverse order:
\begin{equation}
  R_\alpha(W') \le R_\alpha(W)
  \iff
  S_\alpha(W') \ge S_\alpha(W).
  \label{eq:reduction-lt1}
\end{equation}
By \eqref{eq:reduction-lt1}, it suffices to show that $S_\alpha(W') \ge S_\alpha(W)$. 

Using the definition of $S_\alpha$ from \eqref{eq:surrogate-def-lt1}, we can write:
\begin{equation}
  S_\alpha(W) = \max_{Q \in \mathcal{P}(\mathcal{Y})} \min_{x \in \mathcal{X}} \sum_{y} W(y\mid x)^{\alpha} Q(y)^{1-\alpha} .
  \label{eq:surrogate-lt1}
\end{equation}
To facilitate the analysis of the inner minimization in \eqref{eq:surrogate-lt1}, we define the auxiliary function $\Psi_W$ as follows:
\begin{equation}
  \Psi_W(Q) \coloneqq \min_{x\in \mathcal{X}} \sum_{y} W(y\mid x)^{\alpha} Q(y)^{1-\alpha} ,
  \label{eq:augustin-functional-lt1}
\end{equation}
so that $S_\alpha(W) = \max_{Q} \Psi_W(Q)$. We let $\widetilde{Q}$ denote a maximizer of \eqref{eq:augustin-functional-lt1} over the compact simplex $\mathcal{P}(\mathcal{Y})$, referred to as a R\'enyi center of the base channel $W$~\cite{nakiboglu2019renyi}. The existence of $\widetilde{Q}$ follows directly from the extreme value theorem \cite{rudin1976principles}, guaranteed by the continuity of $\Psi_W$ on a compact domain: for $0 < \alpha < 1$, the exponent $1-\alpha$ is positive, so the mapping $t \mapsto t^{1-\alpha}$ is continuous on $[0,\infty)$, each inner summation in \eqref{eq:augustin-functional-lt1} is continuous in $Q$, and the pointwise minimum of finitely many continuous functions is continuous. This ensures the attainment of the maximum value $\Psi_W(\widetilde{Q}) = S_\alpha(W)$. Moreover, since $t \mapsto t^{1-\alpha}$ is also concave, each inner summation is concave in $Q$ and so is their pointwise minimum; hence $\Psi_W$ is concave.

Passing $\widetilde{Q}$ through the marginal $F$ of \eqref{eq:ns2} gives:
\begin{equation}
  \widehat{Q}(y') \coloneqq \sum_{y} F(y'\mid y)\, \widetilde{Q}(y) .
  \label{eq:pushforward-lt1}
\end{equation}
As a mixture of the distributions $F(\cdot\mid y)$, $\widehat{Q}$ is non-negative, and summing over $y'$ yields:
\begin{equation}
  \sum_{y'} \widehat{Q}(y') = \sum_{y'} \sum_{y} F(y'\mid y)\, \widetilde{Q}(y) = \sum_{y} \widetilde{Q}(y) \sum_{y'} F(y'\mid y) = \sum_{y} \widetilde{Q}(y) = 1.
  \label{eq:qhat-sum-lt1}
\end{equation}

To bridge the gap between $S_\alpha(W')$ and $S_\alpha(W)$, we introduce the auxiliary quantity $T$. Since $\widehat{Q}$ is a feasible candidate in the maximization defining $S_\alpha(W')$, we obtain the feasibility bound:
\begin{equation}
  S_\alpha(W') \ge T,
  \qquad
  T \coloneqq \min_{x' \in \mathcal{X}'} \sum_{y'} W'(y'\mid x')^{\alpha}\, \widehat{Q}(y')^{1-\alpha} .
  \label{eq:feasibility-bound-lt1}
\end{equation}
It remains to show $T \ge S_\alpha(W)$. Since $T$ is a minimum over $x'$, we fix an arbitrary $x'^{*} \in \mathcal{X}'$ and bound the corresponding sum from below.

To evaluate this sum and connect it back to the base channel, we leverage the variational Legendre representation derived from Young's inequality. Letting $\beta \coloneqq \alpha/(\alpha-1)$ be the H\"older conjugate of $\alpha$ (so $\beta < 0$, $\frac{1}{\alpha} + \frac{1}{\beta} = 1$ and $-\frac{1}{\beta} > 0$), the reversed form of Young's inequality \cite{hardy1952inequalities} states that for all $p, q \ge 0$ and $t > 0$,
\begin{equation}
  p\,t - \frac{1}{\beta} q\, t^{\beta} \ \ge\ \frac{1}{\alpha} p^{\alpha} q^{1-\alpha},
  \label{eq:young-lt1}
\end{equation}
with equality when $p, q > 0$ at $t = (p/q)^{\alpha-1}$. 

Applying \eqref{eq:young-lt1} pointwise and summing implies that for any conditional distribution $ W(\cdot\mid x)$ and any center $Q$,
\begin{equation}
  \frac{1}{\alpha} \sum_{y} W(y\mid x)^{\alpha} Q(y)^{1-\alpha} = \inf_{h > 0} \left[ \sum_{y} W(y\mid x) h(y) - \frac{1}{\beta} \sum_{y} Q(y) h(y)^{\beta} \right],
  \label{eq:legendre-lt1}
\end{equation}
where, whenever $W(y\mid x)$ and $Q(y)$ are positive, the infimum over $h(y)$ is attained at $h(y) = (W(y\mid x)/Q(y))^{\alpha-1}$.

To connect $S_\alpha(W)$ with $T$ using the Legendre representation, based on an arbitrary test function $f > 0$, we define the local test function $\widehat{f}_{x}(y)$ as follows:
\begin{equation}
  \widehat{f}_{x}(y) \coloneqq \frac{1}{E(x\mid x'^{*})} \sum_{y'} \Pi(x, y' \mid x'^{*}, y) f(y'),
  \label{eq:favg-lt1}
\end{equation}
where, by the no-backward signaling condition \eqref{eq:ns1}, $\Pi(x, y' \mid x'^{*}, y) / E(x\mid x'^{*})$ forms a probability distribution in $y'$. Therefore, $\widehat{f}_{x}(y)$ represents the expected value of $f$ under this conditional distribution:
\begin{equation}
  \widehat{f}_{x}(y) = \mathbb{E}_{y'\sim \Pi(x, y'\mid x'^{*}, y)/E(x\mid x'^{*})} [f(y')],
  \label{eq:favg-expectation-lt1}
\end{equation}
which is a convex combination of the values of $f$\footnote{If $E(x\mid x'^{*}) = 0$ for some $x$, then $\Pi(x, y' \mid x'^{*}, y) = 0$ for all $y, y'$, so such $x$ contribute nothing to any of the sums below and may be omitted.}.

Using the Legendre representation for the conditional distribution $ W(\cdot\mid x)$ and the center $\widetilde{Q}$, by substituting the test function $\widehat{f}_{x}$ in place of the general test function, we obtain the inequality:
\begin{equation}
\begin{aligned}
  \frac{1}{\alpha} \sum_{y} W(y\mid x)^{\alpha} \widetilde{Q}(y)^{1-\alpha} 
  &= \inf_{h > 0} \left[ \sum_{y} W(y\mid x) h(y) - \frac{1}{\beta} \sum_{y} \widetilde{Q}(y) h(y)^{\beta} \right] \\
  &\le \left[ \sum_{y} W(y\mid x) \widehat{f}_{x}(y) - \frac{1}{\beta} \sum_{y} \widetilde{Q}(y) \widehat{f}_{x}(y)^{\beta} \right].
\end{aligned}
\label{eq:legendre-P-lt1}
\end{equation}

Multiplying both sides of the inequality \eqref{eq:legendre-P-lt1} by $E(x\mid x'^{*})$ and summing over $x$, we obtain:
\begin{equation}
\begin{split}
  \frac{1}{\alpha} \sum_{x} E(x\mid x'^{*}) \sum_{y} W(y\mid x)^{\alpha} \widetilde{Q}(y)^{1-\alpha} 
  &\le \sum_{x} E(x\mid x'^{*}) \sum_{y} W(y\mid x) \widehat{f}_{x}(y) \\
  &\quad - \frac{1}{\beta} \sum_{x} E(x\mid x'^{*}) \sum_{y} \widetilde{Q}(y) \widehat{f}_{x}(y)^{\beta}.
\end{split}
\label{eq:multiplied-lt1}
\end{equation}

On the right-hand side of \eqref{eq:multiplied-lt1}, there is a linear term and a nonlinear term, and we handle each of them separately.

For the linear term, expanding $\widehat{f}_{x}$ via \eqref{eq:favg-lt1} and then applying the simulation identity \eqref{eq:ns-simulation}, we obtain:
\begin{equation}
\begin{split}
  \sum_{x} E(x\mid x'^{*}) \sum_{y} W(y\mid x)\, \widehat{f}_{x}(y) 
  &= \sum_{y'} f(y') \sum_{x,y} W(y\mid x)  \Pi(x, y' \mid x'^{*}, y) \\
  &= \sum_{y'} f(y') W'(y'\mid x'^{*}).
\end{split}
\label{eq:linear-term-lt1}
\end{equation}

For the nonlinear term, we seek an upper bound. Since $\beta < 0$, the function $t \mapsto t^{\beta}$ is convex, and Jensen's inequality \cite{hardy1952inequalities} applied to the average $\widehat{f}_{x}(y)$ in \eqref{eq:favg-expectation-lt1} gives:
\begin{equation}
  \widehat{f}_{x}(y)^{\beta} \le \frac{1}{E(x\mid x'^{*})} \sum_{y'} \Pi(x, y' \mid x'^{*}, y) f(y')^{\beta}.
  \label{eq:jensen-ineq-lt1}
\end{equation}
Multiplying by $E(x \mid x'^{*})\widetilde{Q}(y)$, summing over $x$ and $y$, using the no-forward signaling condition \eqref{eq:ns2} to collapse the $x$ sum, and recognizing the remaining sum over $y$ as the pushforward \eqref{eq:pushforward-lt1}, we obtain:
\begin{equation}
  \sum_{x} E(x \mid x'^{*}) \sum_{y} \widetilde{Q}(y)\, \widehat{f}_{x}(y)^{\beta} \le \sum_{y'} f(y')^{\beta} \sum_{y} F(y'\mid y)\, \widetilde{Q}(y) = \sum_{y'} f(y')^{\beta}\, \widehat{Q}(y').
  \label{eq:jensen-collapsed-lt1}
\end{equation}
Substituting the identity \eqref{eq:linear-term-lt1} for the linear term and the upper bound \eqref{eq:jensen-collapsed-lt1} for the nonlinear term into \eqref{eq:multiplied-lt1}, where the latter substitution preserves the direction of the inequality since $-1/\beta > 0$, we obtain:
\begin{equation}
  \frac{1}{\alpha} \sum_{x} E(x\mid x'^{*}) \sum_{y} W(y\mid x)^{\alpha} \widetilde{Q}(y)^{1-\alpha} \le \sum_{y'} W'(y'\mid x'^{*})\, f(y') - \frac{1}{\beta} \sum_{y'} \widehat{Q}(y')\, f(y')^{\beta}.
  \label{eq:collapse-lt1}
\end{equation}
In \eqref{eq:collapse-lt1}, only the right-hand side depends on $f$, so we may optimize it over $f > 0$. To evaluate its infimum, we apply the Legendre representation \eqref{eq:legendre-lt1} to the channel $W'(\cdot\mid x'^{*})$ with center $\widehat{Q}$, which reads:
\begin{equation}
  \frac{1}{\alpha} \sum_{y'} W'(y'\mid x'^{*})^{\alpha}\, \widehat{Q}(y')^{1-\alpha} = \inf_{f > 0} \left[ \sum_{y'} W'(y'\mid x'^{*})\, f(y') - \frac{1}{\beta} \sum_{y'} \widehat{Q}(y')\, f(y')^{\beta} \right].
  \label{eq:T-legendre-lt1}
\end{equation}
Taking the infimum over $f > 0$ in \eqref{eq:collapse-lt1} and substituting \eqref{eq:T-legendre-lt1}, we multiply both sides by $\alpha > 0$ to obtain:
\begin{equation}
  \sum_{x} E(x\mid x'^{*}) \sum_{y} W(y\mid x)^{\alpha}\, \widetilde{Q}(y)^{1-\alpha} \le \sum_{y'} W'(y'\mid x'^{*})^{\alpha}\, \widehat{Q}(y')^{1-\alpha}.
  \label {eq:infimum-taken-lt1}
\end{equation}
Since $E(\cdot \mid x'^{*})$ is a probability distribution over $x$ owing to the no-backward signaling condition \eqref{eq:ns1}, an average over $x$ is bounded from below by the minimum over $x$. Consequently, we have:
\begin{equation}
  S_\alpha(W) = \min_{x} \sum_{y} W(y\mid x)^{\alpha} \widetilde{Q}(y)^{1-\alpha} \le \sum_{x} E(x\mid x'^{*}) \sum_{y} W(y\mid x)^{\alpha} \widetilde{Q}(y)^{1-\alpha},
  \label{eq:lower-bound-lt1}
\end{equation}
where the first equality holds by \eqref{eq:augustin-functional-lt1}, since $\widetilde{Q}$ attains the maximum of $\Psi_W$.

By combining \eqref{eq:infimum-taken-lt1} and \eqref{eq:lower-bound-lt1}, we obtain:
\begin{equation}
  S_\alpha(W) \le \sum_{y'} W'(y'\mid x'^{*})^{\alpha}\, \widehat{Q}(y')^{1-\alpha}.
  \label{eq:final-bound-lt1}
\end{equation}
Since only the right-hand side of \eqref{eq:final-bound-lt1} depends on $x'^{*}$, which was chosen arbitrarily from the beginning, we can take the minimum over $x'^{*}$ on the right-hand side. Consequently, we have:
\begin{equation}
  S_\alpha(W) \le \min_{x' \in \mathcal{X}'} \sum_{y'} W'(y'\mid x')^{\alpha}\, \widehat{Q}(y')^{1-\alpha} = T,
  \label{eq:renyi-capacity-bound-lt1}
\end{equation}
where the last equality holds by the definition of $T$ \eqref{eq:feasibility-bound-lt1}.

Therefore, combining \eqref{eq:renyi-capacity-bound-lt1} and \eqref{eq:feasibility-bound-lt1}, we obtain:
\begin{equation}
  S_\alpha(W) \le T \le S_\alpha(W'),
  \label{eq:monotonicity-chain-lt1}
\end{equation}
which, by \eqref{eq:reduction-lt1}, completes the proof.
\end{proof}

\subsection{Monotonicity for \texorpdfstring{$\alpha>1$}{alpha>1}}
\label{sec:alphagt1}

\begin{lemma}
\label{lem:alphagt1}
If the simulator channel $W'$ is non-signaling simulable from the base channel $W$, then their respective R\'enyi capacities of order $\alpha > 1$ satisfy the monotonicity relation:
\begin{equation}
  R_\alpha(W') \le R_\alpha(W).
  \label{eq:lemma-alphagt1}
\end{equation}
\end{lemma}

\begin{proof}
The proof is overall very similar to the previous case, $0<\alpha<1$, but because of the sign as well as convexity/concavity differences we write them separately. 

Throughout the proof, $\alpha \in (1, \infty)$ is held fixed. To eliminate logarithms and simplify the analytical expressions, we introduce the exponentiated surrogate as follows:
\begin{equation}
  S_\alpha(W) \coloneqq e^{(\alpha-1) R_\alpha(W)} .
  \label{eq:surrogate-def}
\end{equation}
Since $\alpha > 1$, $e^{(\alpha-1)(\cdot)}$ is strictly increasing, so comparing $R_\alpha$ is equivalent to comparing $S_\alpha$:
\begin{equation}
  R_\alpha(W') \le R_\alpha(W)
  \iff
  S_\alpha(W') \le S_\alpha(W).
  \label{eq:reduction-gt1}
\end{equation}
By \eqref{eq:reduction-gt1}, it suffices to show that $S_\alpha(W') \le S_\alpha(W)$. Using the definition of $S_\alpha$ from \eqref{eq:surrogate-def}, we can write:
\begin{equation}
  S_\alpha(W) = \min_{Q \in \mathcal{P}(\mathcal{Y})} \max_{x \in \mathcal{X}} \sum_{y} W(y\mid x)^{\alpha} Q(y)^{1-\alpha} .
  \label{eq:surrogate-gt1}
\end{equation}
To facilitate the analysis of the inner maximization in \eqref{eq:surrogate-gt1}, we define the auxiliary function $\Phi_W$ as follows:
\begin{equation}
  \Phi_W(Q) \coloneqq \max_{x\in \mathcal{X}}\sum_{y} W(y\mid x)^{\alpha} Q(y)^{1-\alpha} ,
  \label{eq:augustin-functional-def}
\end{equation}
so that $S_\alpha(W) = \min_{Q} \Phi_W(Q)$. We let $\widetilde{Q}$ denote a minimizer of \eqref{eq:augustin-functional-def} over the compact simplex $\mathcal{P}(\mathcal{Y})$, referred to as a R\'enyi center of the base channel $W$~\cite{nakiboglu2019renyi}. The existence of $\widetilde{Q}$ follows directly from the extreme value theorem \cite{rudin1976principles}, guaranteed by the continuity of $\Phi_W$ on a compact domain: for $\alpha > 1$, the exponent $1-\alpha$ is negative, so each summand $W(y\mid x)^{\alpha} Q(y)^{1-\alpha}$ is continuous in $Q$ as a map into $[0,\infty]$, diverging to $+\infty$ as $Q(y) \to 0$ whenever $W(y\mid x) > 0$, and the pointwise maximum of finitely many such maps is continuous. This ensures the attainment of the minimum value $\Phi_W(\widetilde{Q}) = S_\alpha(W)$, which is finite by the fact that $R_\alpha(W) < \infty$ for finite alphabets. Moreover, since $t\mapsto t^{1-\alpha}$ is convex, $\Phi_W$ is convex.

An essential property of this minimizer is that $\widetilde{Q}(y) > 0$ whenever $W(y\mid x) > 0$ for some $x$. To see this, suppose to the contrary that there exists some $y_0$ such that $W(y_0\mid x) > 0$ for a given $x$ while $\widetilde{Q}(y_0) = 0$. Since $\alpha > 1$, the exponent $1-\alpha$ is negative, forcing the corresponding term in \eqref{eq:augustin-functional-def} to diverge to $+\infty$. This would imply $\Phi_W(\widetilde{Q}) = +\infty$, contradicting the finiteness of $S_\alpha(W)$. Thus, the support of $\widetilde{Q}$ must cover all transitions allowed by $W$.

Passing $\widetilde{Q}$ through the marginal $F$ of \eqref{eq:ns2} gives:
\begin{equation}
  \widehat{Q}(y') \coloneqq \sum_{y} F(y'\mid y)\, \widetilde{Q}(y) .
  \label{eq:pushforward-gt1}
\end{equation}
As a mixture of the distributions $F(\cdot\mid y)$, $\widehat{Q}$ is a probability distribution on $\mathcal{Y}'$, since summing over $y'$ yields:
\begin{equation}
  \sum_{y'} \widehat{Q}(y') = \sum_{y'} \sum_{y} F(y'\mid y)\, \widetilde{Q}(y) = \sum_{y} \widetilde{Q}(y) \sum_{y'} F(y'\mid y) = \sum_{y} \widetilde{Q}(y) = 1.
  \label{eq:qhat-sum}
\end{equation}
Furthermore, $\widehat{Q}(y') > 0$ whenever $W'(y'\mid x') > 0$ for some $x'$. To verify this support property, suppose that $W'(y'\mid x') > 0$ for a given input $x'$. By the non-signaling simulation definition \eqref{eq:ns-simulation}, the conditional probability is expanded as:
\begin{equation}
  W'(y'\mid x') = \sum_{x,y} \Pi(x, y' \mid x', y)\, W(y\mid x) > 0.
  \label{eq:support-sim-expand}
\end{equation}
For this sum to be strictly positive, there must exist at least one pair $(x,y)$ in the support of the summation such that $\Pi(x, y' \mid x', y) > 0$ and $W(y\mid x) > 0$. 
Given that $W(y\mid x) > 0$, the previously established property of the R\'enyi center $\widetilde{Q}$ guarantees that $\widetilde{Q}(y) > 0$. Furthermore, invoking the no-forward signaling condition \eqref{eq:ns2}, the marginal decoder satisfies:
\begin{equation}
  F(y'\mid y) = \sum_{x''} \Pi(x'', y' \mid x', y) \ge \Pi(x, y' \mid x', y) > 0.
  \label{eq:support-ns2}
\end{equation}
Substituting these strictly positive bounds into the definition of $\widehat{Q}$ in \eqref{eq:pushforward-gt1}, and lower-bounding the full summation by the specific term corresponding to $y$, we obtain:
\begin{equation}
  \widehat{Q}(y') = \sum_{y''} F(y'\mid y'')\, \widetilde{Q}(y'') \ge F(y'\mid y)\, \widetilde{Q}(y) > 0.
  \label{eq:qhat-support}
\end{equation}
Hence $\widehat{Q}$ is strictly positive on the support of $W'$.

Having constructed $\widehat{Q}$, we return to the comparison between $S_\alpha(W')$ and $S_\alpha(W)$. To bridge the gap between them we introduce the auxiliary quantity $T$. Since $\widehat{Q}$ is a feasible candidate in the minimization defining $S_\alpha(W')$, we obtain the feasibility bound:
\begin{equation}
  S_\alpha(W') \le T,
  \qquad
  T \coloneqq \max_{x' \in \mathcal{X}'} \sum_{y'} W'(y'\mid x')^{\alpha}\, \widehat{Q}(y')^{1-\alpha} .
  \label{eq:feasibility-bound}
\end{equation}
Every term is finite by the support property of $\widehat{Q}$, so $T < \infty$. It remains to show $T \le S_\alpha(W)$. Since $T$ is a maximum over $x'$, we fix an arbitrary $x'^{*} \in \mathcal{X}'$ and bound the corresponding sum from above.

To evaluate this sum and connect it back to the base channel, we leverage the variational (Legendre) representation derived from Young's inequality. Letting $\beta \coloneqq \alpha/(\alpha-1)$ be the H\"older conjugate of $\alpha$ (so $\beta > 1$ and $\frac{1}{\alpha} + \frac{1}{\beta} = 1$), Young's inequality \cite{hardy1952inequalities} states that for all $p, q \ge 0$ and $t \ge 0$,
\begin{equation}
  p\,t - \frac{1}{\beta} q\, t^{\beta} \ \le\ \frac{1}{\alpha} p^{\alpha} q^{1-\alpha},
  \label{eq:young-gt1}
\end{equation}
with equality when $p, q > 0$ at $t = (p/q)^{\alpha-1}$. 

Applying \eqref{eq:young-gt1} pointwise and summing implies that for any conditional distribution $W(\cdot\mid x)$ and any center $Q$,
\begin{equation}
  \frac{1}{\alpha} \sum_{y} W(y\mid x)^{\alpha} Q(y)^{1-\alpha} = \sup_{h \ge 0} \left[ \sum_{y} W(y\mid x) h(y) - \frac{1}{\beta} \sum_{y} Q(y) h(y)^{\beta} \right],
  \label{eq:legendre}
\end{equation}
where, whenever $Q(y) > 0$, the supremum over $h(y)$ is attained at $h(y) = (W(y\mid x)/Q(y))^{\alpha-1}$.

To connect $T$ with $S_\alpha(W)$ using the Legendre representation, based on an arbitrary test function $f \ge 0$, we define the local test function $\widehat{f}_{x}(y)$ as follows:
\begin{equation}
  \widehat{f}_{x}(y) \coloneqq \frac{1}{E(x\mid x'^{*})} \sum_{y'} \Pi(x, y' \mid x'^{*}, y) f(y'),
  \label{eq:favg}
\end{equation}
where, by the no-backward signaling condition \eqref{eq:ns1}, $\Pi(x, y' \mid x'^{*}, y) / E(x\mid x'^{*})$ forms a probability distribution in $y'$. Therefore, $\widehat{f}_{x}(y)$ represents the expected value of $f$ under this conditional distribution:
\begin{equation}
  \widehat{f}_{x}(y) = \mathbb{E}_{y'\sim \Pi(x, y'\mid x'^{*}, y)/E(x\mid x'^{*})} [f(y')],
  \label{eq:favg-expectation}
\end{equation}
which is a convex combination of the values of $f$ \footnote{If $E(x\mid x'^{*}) = 0$ for some $x$, then $\Pi(x, y' \mid x'^{*}, y) = 0$ for all $y, y'$, so such $x$ contribute nothing to any of the sums below and may be omitted.}.

Using the Legendre representation for the conditional distribution $W(\cdot\mid x)$ and the center $\widetilde{Q}$, by substituting the test function $\widehat{f}_{x}$ in place of the general test function, we obtain the inequality:
\begin{equation}
\begin{aligned}
  \frac{1}{\alpha} \sum_{y} W(y\mid x)^{\alpha} \widetilde{Q}(y)^{1-\alpha}
  &= \sup_{h \ge 0} \left[ \sum_{y} W(y\mid x) h(y) - \frac{1}{\beta} \sum_{y} \widetilde{Q}(y) h(y)^{\beta} \right] \\
  &\ge \left[ \sum_{y} W(y\mid x) \widehat{f}_{x}(y) - \frac{1}{\beta} \sum_{y} \widetilde{Q}(y) \widehat{f}_{x}(y)^{\beta} \right].
\end{aligned}
\label{eq:legendre-P}
\end{equation}

Multiplying both sides of the inequality \eqref{eq:legendre-P} by $E(x\mid x'^{*})$ and summing over $x$, we obtain:
\begin{equation}
\begin{split}
  \frac{1}{\alpha} \sum_{x} E(x\mid x'^{*}) \sum_{y} W(y\mid x)^{\alpha} \widetilde{Q}(y)^{1-\alpha}
  &\ge \sum_{x} E(x\mid x'^{*}) \sum_{y} W(y\mid x) \widehat{f}_{x}(y) \\
  &\quad - \frac{1}{\beta} \sum_{x} E(x\mid x'^{*}) \sum_{y} \widetilde{Q}(y) \widehat{f}_{x}(y)^{\beta}.
\end{split}
\label{eq:multiplied}
\end{equation}
On the right-hand side of \eqref{eq:multiplied}, there is a linear term and a nonlinear term, and we handle each of them separately.

For the linear term, expanding $\widehat{f}_{x}$ via \eqref{eq:favg} and then applying the simulation identity \eqref{eq:ns-simulation}, we obtain:
\begin{equation}
\begin{split}
  \sum_{x} E(x\mid x'^{*}) \sum_{y} W(y\mid x)\, \widehat{f}_{x}(y)
  &= \sum_{y'} f(y') \sum_{x,y} W(y\mid x)\, \Pi(x, y' \mid x'^{*}, y) \\
  &= \sum_{y'} f(y')\, W'(y'\mid x'^{*}).
\end{split}
\label{eq:linear-term}
\end{equation}

For the nonlinear term, we seek an upper bound. Since $\beta > 1$, the function $t \mapsto t^{\beta}$ is convex, and Jensen's inequality \cite{hardy1952inequalities} applied to the average $\widehat{f}_{x}(y)$ in \eqref{eq:favg-expectation} gives:
\begin{equation}
  \widehat{f}_{x}(y)^{\beta} \le \frac{1}{E(x\mid x'^{*})} \sum_{y'} \Pi(x, y' \mid x'^{*}, y) f(y')^{\beta}.
  \label{eq:jensen-ineq}
\end{equation}
Multiplying by $E(x \mid x'^{*})\widetilde{Q}(y)$, summing over $x$ and $y$, using the no-forward signaling condition \eqref{eq:ns2} to collapse the $x$ sum, and recognizing the remaining sum over $y$ as the pushforward \eqref{eq:pushforward-gt1}, we obtain:
\begin{equation}
  \sum_{x} E(x \mid x'^{*}) \sum_{y} \widetilde{Q}(y)\, \widehat{f}_{x}(y)^{\beta} \le \sum_{y'} f(y')^{\beta} \sum_{y} F(y'\mid y)\, \widetilde{Q}(y) = \sum_{y'} f(y')^{\beta}\, \widehat{Q}(y').
  \label{eq:jensen-collapsed}
\end{equation}
Substituting the identity \eqref{eq:linear-term} for the linear term and the upper bound \eqref{eq:jensen-collapsed} for the nonlinear term into \eqref{eq:multiplied}, where the latter substitution preserves the direction of the inequality since $-1/\beta < 0$, we obtain:
\begin{equation}
  \frac{1}{\alpha} \sum_{x} E(x\mid x'^{*}) \sum_{y} W(y\mid x)^{\alpha} \widetilde{Q}(y)^{1-\alpha} \ge \sum_{y'} W'(y'\mid x'^{*})\, f(y') - \frac{1}{\beta} \sum_{y'} \widehat{Q}(y')\, f(y')^{\beta}.
  \label{eq:collapse}
\end{equation}
In \eqref{eq:collapse}, only the right-hand side depends on $f$, so we may optimize it over $f \ge 0$. To evaluate its supremum, we apply the Legendre representation \eqref{eq:legendre} to the channel $W'(\cdot\mid x'^{*})$ with center $\widehat{Q}$, which reads:
\begin{equation}
  \frac{1}{\alpha} \sum_{y'} W'(y'\mid x'^{*})^{\alpha}\, \widehat{Q}(y')^{1-\alpha} = \sup_{f \ge 0} \left[ \sum_{y'} W'(y'\mid x'^{*})\, f(y') - \frac{1}{\beta} \sum_{y'} \widehat{Q}(y')\, f(y')^{\beta} \right].
  \label{eq:T-legendre}
\end{equation}
Taking the supremum over $f \ge 0$ in \eqref{eq:collapse} and substituting \eqref{eq:T-legendre}, we multiply both sides by $\alpha > 0$ to obtain:
\begin{equation}
  \sum_{x} E(x\mid x'^{*}) \sum_{y} W(y\mid x)^{\alpha}\, \widetilde{Q}(y)^{1-\alpha} \ge \sum_{y'} W'(y'\mid x'^{*})^{\alpha}\, \widehat{Q}(y')^{1-\alpha}.
  \label{eq:supremum-taken}
\end{equation}
Since $E(\cdot \mid x'^{*})$ is a probability distribution over $x$ owing to the no-backward signaling condition \eqref{eq:ns1}, an average over $x$ is bounded from above by the maximum over $x$. Consequently, we have:
\begin{equation}
  \sum_{x} E(x\mid x'^{*}) \sum_{y} W(y\mid x)^{\alpha} \widetilde{Q}(y)^{1-\alpha} \le \max_{x} \sum_{y} W(y\mid x)^{\alpha} \widetilde{Q}(y)^{1-\alpha} = S_\alpha(W),
  \label{eq:upper-bound}
\end{equation}
where the last equality holds by \eqref{eq:augustin-functional-def}, since $\widetilde{Q}$ attains the minimum of $\Phi_W$.

By combining \eqref{eq:supremum-taken} and \eqref{eq:upper-bound}, we obtain:
\begin{equation}
  \sum_{y'} W'(y'\mid x'^{*})^{\alpha}\, \widehat{Q}(y')^{1-\alpha} \le S_\alpha(W).
  \label{eq:final-bound}
\end{equation}
Since only the left-hand side of \eqref{eq:final-bound} depends on $x'^{*}$, which was chosen arbitrarily from the beginning, we can take the maximum over $x'^{*}$ on the left-hand side. Consequently, we have:
\begin{equation}
  T = \max_{x' \in \mathcal{X}'} \sum_{y'} W'(y'\mid x')^{\alpha}\, \widehat{Q}(y')^{1-\alpha} \le S_\alpha(W),
  \label{eq:renyi-capacity-bound}
\end{equation}
where the first equality holds by the definition of $T$ in \eqref{eq:feasibility-bound}.

Therefore, combining \eqref{eq:renyi-capacity-bound} and \eqref{eq:feasibility-bound}, we obtain:
\begin{equation}
  S_\alpha(W') \le T \le S_\alpha(W),
  \label{eq:monotonicity-chain}
\end{equation}
which, by \eqref{eq:reduction-gt1}, completes the proof.
\end{proof}

\subsection{Monotonicity for \texorpdfstring{$\alpha=\infty$}{alpha=infinity}}
\label{sec:alphainf}

\begin{lemma}
\label{lem:alphainf}
If the simulator channel $W'$ is non-signaling simulable from the base channel $W$, then their respective R\'enyi capacities of order $\alpha = \infty$ satisfy the monotonicity relation:
\begin{equation}
  R_\infty(W') \le R_\infty(W).
  \label{eq:lemma-alphainf}
\end{equation}
\end{lemma}

\begin{proof}
At the limiting order $\alpha = \infty$, the information radius admits a closed form, and we establish this first. Writing each summand of \eqref{eq:renyi-divergence} as $P(x)^{\alpha} Q(x)^{1-\alpha} = Q(x) \big(P(x)/Q(x)\big)^{\alpha}$, the sum is a weighted sum of exponentials in $\alpha$ and is therefore dominated, as $\alpha \to \infty$, by the largest likelihood ratio. Consequently,
\begin{equation}
  D_\infty(P \,\|\, Q) \coloneqq \lim_{\alpha \to \infty} D_\alpha(P \,\|\, Q) = \log \max_{x:\,P(x)>0} \frac{P(x)}{Q(x)},
  \label{eq:dinf}
\end{equation}
with the convention $P(x)/Q(x) = +\infty$ whenever $Q(x) = 0 < P(x)$, so that the finiteness of \eqref{eq:dinf} forces $Q$ to dominate $P$.

Substituting \eqref{eq:dinf} into the definition \eqref{eq:info-radius} of the information radius of a channel $W(y\mid x)$ and using the fact that logarithm is increasing, the two inner maxima merge into a joint maximum over $x$ and $y$, and we obtain:
\begin{equation}
  R_\infty(W) = \min_{Q} \max_{x} \log \max_{y} \frac{W(y\mid x)}{Q(y)} = \log \min_{Q} \max_{y} \frac{m(y)}{Q(y)}, \qquad m(y) \coloneqq \max_{x \in \mathcal{X}} W(y\mid x).
  \label{eq:rinf-minimax}
\end{equation}
The remaining minimization runs over the simplex $\mathcal{P}(\mathcal{Y})$. Its optimum is attained when the ratios $m(y)/Q(y)$ are equalized across all $y$ with $m(y) > 0$, since any $Q$ placing less mass on such a $y$ can always be improved by moving mass toward it to reduce the maximum ratio. Precisely, for every $Q$, summing over $y$ with $m(y)>0$,
\begin{equation}
    \max_y m(y)/Q(y) \ge \sum_y m(y)\big/\sum_y Q(y) \ge \sum_y m(y),
\end{equation} with equality for the equalizing distribution.
 Therefore, the equalizing distribution is:
\begin{equation}
  Q^{*}(y) = \frac{m(y)}{\sum_{y} m(y)}.
  \label{eq:qstar}
\end{equation}
Substituting \eqref{eq:qstar} into \eqref{eq:rinf-minimax}, the information radius at order $\alpha = \infty$ is
\begin{equation}
  R_\infty(W) = \log \sum_{y \in \mathcal{Y}} \max_{x \in \mathcal{X}} W(y\mid x) \eqqcolon \log \lVert M_W \rVert_1 ,
  \label{eq:rinf-closed-form}
\end{equation}
where the shorthand notation $\lVert M_W \rVert_1$ is adopted. Since the logarithm is a strictly increasing function, the claim \eqref{eq:lemma-alphainf} is equivalent to the inequality $\lVert M_{W'} \rVert_1 \le \lVert M_W \rVert_1$, which is established as follows.

The pointwise maximum in \eqref{eq:rinf-closed-form} is addressed by first linearizing the expression: for any channel $W(y\mid x)$,
\begin{equation}
  \sum_{y \in \mathcal{Y}} \max_{x \in \mathcal{X}} W(y\mid x) = \max_{\sigma} \sum_{x,y} \sigma(x\mid y)\, W(y\mid x),
  \label{eq:variational}
\end{equation}
where the maximum runs over all conditional distributions $\sigma(x \mid y)$. The right-hand side is linear in the channel, and we substitute \eqref{eq:ns-simulation} into it.
Let $\sigma^{*}$ attain the maximum in \eqref{eq:variational} for the simulator channel $W'$. Substituting the simulation identity \eqref{eq:ns-simulation} and exchanging the order of summation yields:
\begin{align}
  \lVert M_{W'} \rVert_1
    &= \sum_{x',y'} \sigma^{*}(x'\mid y')\, W'(y'\mid x') \label{eq:inf-step-1} \\
    &= \sum_{x',y'} \sigma^{*}(x'\mid y') \sum_{x,y} \Pi(x,y'\mid x',y)\, W(y\mid x) \label{eq:inf-step-2} \\
    &= \sum_{x,y} W(y\mid x)\, A(x,y), \label{eq:inf-step-3}
\end{align}
where we have introduced:
\begin{equation}
  A(x,y) \coloneqq \sum_{x',y'} \sigma^{*}(x'\mid y')\, \Pi(x,y'\mid x',y).
  \label{eq:adef}
\end{equation}
For each fixed $y$, the quantity $A(\cdot,y)$ defines a probability distribution over $x$, as we now verify. Non-negativity is immediate from \eqref{eq:adef}. For the normalization, summing over $x$ and invoking the no-forward signaling condition \eqref{eq:ns2} yields:
\begin{align}
  \sum_{x} A(x,y)
    &= \sum_{x',y'} \sigma^{*}(x'\mid y') \sum_{x} \Pi(x,y'\mid x',y) \label{eq:inf-norm-1} \\
    &= \sum_{x',y'} \sigma^{*}(x'\mid y')\, F(y'\mid y) \label{eq:inf-norm-2} \\
    &= \sum_{y'} F(y'\mid y) \sum_{x'} \sigma^{*}(x'\mid y') \label{eq:inf-norm-3} \\
    &= \sum_{y'} F(y'\mid y) \label{eq:inf-norm-4} \\
    &= 1. \label{eq:inf-norm-5}
\end{align}

Consequently, for each fixed $y$ the inner sum in \eqref{eq:inf-step-3} is an average of the values $W(y\mid x)$ over $x$, and is therefore bounded above by the largest of them:
\begin{equation}
  \sum_{x} A(x,y)\, W(y\mid x) \le \max_{x} W(y\mid x).
  \label{eq:average-le-max}
\end{equation}
Recalling \eqref{eq:inf-step-3}, which amounts to summing
\eqref{eq:average-le-max} over $y$, and substituting the closed form
\eqref{eq:rinf-closed-form}, we obtain:
\begin{equation}
  \lVert M_{W'} \rVert_1 = \sum_{y} \sum_{x} A(x,y)\, W(y\mid x) \le \sum_{y} \max_{x} W(y\mid x) = \lVert M_W \rVert_1 .
  \label{eq:inf-final}
\end{equation}

By \eqref{eq:rinf-closed-form}, the inequality $\lVert M_{W'} \rVert_1 \le
\lVert M_W \rVert_1$ established in \eqref{eq:inf-final} is equivalent to
$R_\infty(W') \le R_\infty(W)$. This is the claim
\eqref{eq:lemma-alphainf}, which completes the proof of monotonicity for
$\alpha = \infty$.
\end{proof}

\begin{remark}\label{rem:only-ns2} The proof relies exclusively on the no-forward signaling condition \eqref{eq:ns2} in step \eqref{eq:inf-norm-2}, whereas the no-backward signaling condition \eqref{eq:ns1} is never invoked. Since the linearization \eqref{eq:variational} reduces the problem to the outer-output marginal governed by \eqref{eq:ns2}, monotonicity of the R\'enyi capacity at order $\alpha = \infty$ holds for the strictly larger class of boxes satisfying \eqref{eq:ns2} alone, provided that \eqref{eq:ns-simulation} still defines a valid channel $W'$.
\end{remark}

\subsection{The endpoint $\alpha \to 0$ and zero-error capacity}
\label{sec:disc-zero}
 
So far $\alpha \in (0, \infty]$ has been covered. The endpoint
$\alpha = 0$ has its own operational meaning. The R\'enyi divergence of order $0$
is $D_0(P \,\|\, Q) = -\log Q(\operatorname{supp} P)$~\cite{vanerven2014}, so
the information radius \eqref{eq:info-radius} at order $0$ is:
\begin{equation}\label{eq:R0}
  R_0(W) = \min_{Q \in \mathcal{P}(\mathcal{Y})} \max_{x \in \mathcal{X}}
  \big( -\log Q(\mathcal{S}_x) \big)
  = -\log \max_{Q \in \mathcal{P}(\mathcal{Y})} \min_{x \in \mathcal{X}}
  Q(\mathcal{S}_x),
  \qquad
  \mathcal{S}_x := \{ y : W(y \mid x) > 0 \}.
\end{equation}
The optimization is a linear program, and linear programming duality \cite{schrijver1986,scheinerman1997} gives
\begin{equation}\label{eq:R0-packing}
  e^{R_0(W)} = \alpha^*(W) := \max \Big\{ \sum_{x} u(x) \;:\;
  u \ge 0, \ \sum_{x : W(y \mid x) > 0} u(x) \le 1 \ \text{for all } y
  \Big\},
\end{equation}
the fractional packing number of the bipartite graph of $W$, a standard quantity in fractional graph
theory~\cite{scheinerman1997} and the one used
in~\cite{cubitt2011zero}. Cubitt, Leung, Matthews and Winter~\cite{cubitt2011zero} proved that $\log \alpha^*(W)$ is exactly the non-signaling assisted zero-error capacity of $W$. So the zero-error end of the dial is the non-signaling zero-error capacity: non-signaling-assisted codes achieve zero error at every rate below $R_0(W)$. (The same number also appears as Shannon's formula for the zero-error capacity with feedback~\cite{shannon1956}, valid when the unassisted zero-error
capacity is positive.)
 
Lemma~\ref{lem:alphalt1} can be extended to this endpoint directly.

\begin{lemma}
    \label{lem:alphaeq0}
If the simulator channel $W'$ is non-signaling simulable from the base channel $W$, then their respective information radii of order $0$ satisfy the monotonicity relation:
\begin{equation}
  R_0(W') \le R_0(W).
  \label{eq:lemma-alphaeq0}
\end{equation}
\end{lemma}

\begin{proof}
    For finite alphabets, $\alpha \mapsto R_\alpha(W)$ is nondecreasing and $R_\alpha(W) \to R_0(W)$ as $\alpha \downarrow 0$: the upper bound follows by evaluating the radius at a minimizer $Q$ for $R_0(W)$ and using the continuity of $\alpha \mapsto D_\alpha$ on $[0,1]$. Letting $\alpha \downarrow 0$ in Lemma~\ref{lem:alphalt1} gives $R_0(W') \le R_0(W)$, which recovers the monotonicity of the non-signaling zero-error capacity under non-signaling simulation~\cite{cubitt2011zero}.
\end{proof}
 
Comparing with the unassisted problem shows that monotonicity is a
property of the R\'enyi family and not of capacity-like quantities in
general. The unassisted zero-error capacity $C_0(W)$, the largest rate at which communication with zero error is possible, is the logarithm of the Shannon capacity of the confusability graph of $W$, and it can
strictly increase when the sender and receiver share
entanglement~\cite{cubitt2011zero,leung2012}. An entanglement-assisted encoder
and decoder form a non-signaling box, so the unassisted zero-error
capacity is not monotone under non-signaling simulation, while
$R_0$ is.
\subsection{The main Theorem}
\label{sec7main}
\begin{theorem}[{Consequence of~\cite[Lemma~24]{Oufkir2026Strong}}]
\label{thm:main}
If a simulator channel $W'$ is non-signaling simulable from a base channel $W$, then their respective R\'enyi capacities satisfy:
\begin{equation}
  R_\alpha(W') \le R_\alpha(W)
  \qquad \text{for every } \alpha \in [0,\infty].
  \label{eq:main-theorem}
\end{equation}
\end{theorem}

\begin{proof}
Fix $\alpha \in [0, \infty]$. The claim follows from the five regimes established in Subsections~\ref{sec:alphaeq1}--\ref{sec:disc-zero}: Lemma~\ref{lem:alphaeq1} for $\alpha = 1$ (noting $R_1(W) = C(W)$ from \eqref{eq:info-radius}), Lemma~\ref{lem:alphalt1} for $0 < \alpha < 1$, Lemma~\ref{lem:alphagt1} for $1 < \alpha < \infty$, Lemma~\ref{lem:alphainf} for $\alpha = \infty$ and Lemma~\ref{lem:alphaeq0} for $\alpha=0$. These cases collectively cover $[0, \infty]$, completing the proof.
\end{proof}

\begin{remark}
Theorem~\ref{thm:main} also follows from~\cite[Lemma~24]{Oufkir2026Strong}. To see this, let $W'$ be obtained from $W$ through a non-signaling box $\Pi$. For any output distribution $Q$, define the constant-output channel $T_Q(y\mid x)=Q(y)$. The no-forward-signaling condition implies that applying the same box to $T_Q$ gives a constant-output channel with distribution
\[
\widehat Q(y')=\sum_{x,y}\Pi(x,y'\mid x',y)Q(y)=\sum_y F(y'\mid y)Q(y),
\]
which is independent of $x'$. For every $\alpha\in(0,\infty)$, the data-processing inequality of~\cite[Lemma~24]{Oufkir2026Strong}, applied to $W$ and $T_Q$, therefore gives
\[
R_\alpha(W')
\leq \max_{x'}D_\alpha\!\left(W'(\cdot\mid x')\,\middle\|\,\widehat Q\right)
\leq \max_xD_\alpha\!\left(W(\cdot\mid x)\,\middle\|\,Q\right).
\]
Minimizing the right-hand side over $Q$ yields $R_\alpha(W')\leq R_\alpha(W)$. The endpoint orders follow by continuity of the R\'enyi capacity for finite channels~\cite{nakiboglu2019renyi}. The preceding lemmas provide an alternative proof, using a variational representation for $0<\alpha<1$ and $\alpha>1$, together with separate arguments for the Shannon capacity and endpoint orders.
\end{remark}
\section{Applications}
\label{sec:applications}

Theorem~\ref{thm:main} states that a non-signaling box cannot increase any R\'enyi capacity of a classical channel. In this section, we explain what this means for the rate--reliability trade-offs of Section~\ref{sec:intro} and derive bounds on the success probability of assisted codes. As in the previous section, $W'$ is a channel that is non-signaling simulable from $W$ by super-channel $\Pi$ in the sense of Definition~\ref{def:ns-simulation}.

\subsection{Error exponent curves of the simulated channel}
\label{sec:disc-curves}

Gallager's~\cite{gallager1965simple} function of a channel $W$ and an input distribution $P$ is
\begin{equation}\label{eq:E0}
  E_0(\rho,P,W)
  := -\log \sum_y
  \left(\sum_x P(x)W(y\mid x)^{1/(1+\rho)}\right)^{1+\rho},
  \qquad \rho>-1.
\end{equation}
Its relation to the R\'enyi capacity~\cite{csiszar1995,nakiboglu2019renyi} depends on the sign of $\rho$:
\begin{equation}\label{eq:csiszar}
  \rho R_{1/(1+\rho)}(W)
  =
  \begin{cases}
    \displaystyle\max_{P\in\mathcal{P}(\mathcal{X})}
      E_0(\rho,P,W), & \rho>0,\\[6pt]
    \displaystyle\min_{P\in\mathcal{P}(\mathcal{X})}
      E_0(\rho,P,W), & -1<\rho<0.
  \end{cases}
\end{equation}
Indeed, $E_0(\rho,P,W)=\rho I_{1/(1+\rho)}(P,W)$, where $I_\alpha(P,W)$ is the Sibson R\'enyi information and $R_\alpha(W)=\max_P I_\alpha(P,W)$. Multiplication by a negative $\rho$ reverses the optimization. At $\rho=0$, $E_0(0,P,W)=0$ for every $P$. Writing $\alpha = 1/(1+\rho)$, so that $\rho = (1-\alpha)/\alpha$, the three classical exponent curves become functions of the R\'enyi capacities alone:
\begin{align}
  E_{\mathrm{r}}(R, W)
    &:= \sup_{\alpha \in [1/2, 1)} \frac{1-\alpha}{\alpha}
        \big( R_\alpha(W) - R \big), \label{eq:Er} \\
  E_{\mathrm{sp}}(R, W)
    &:= \sup_{\alpha \in (0, 1)} \frac{1-\alpha}{\alpha}
        \big( R_\alpha(W) - R \big), \label{eq:Esp} \\
  E_{\mathrm{sc}}(R, W)
    &:= \sup_{\alpha \in (1, \infty)} \frac{\alpha-1}{\alpha}
        \big( R - R_\alpha(W) \big). \label{eq:Esc}
\end{align}
The random-coding and sphere-packing curves bound the reliability function~\cite{gallager1965simple,shannon1967lower}, an operational quantity that we now define precisely. An $(n, M)$ code for $W$ is an encoder $f : \{1, \dots, M\} \to \mathcal{X}^n$ together with a decoder $g : \mathcal{Y}^n \to \{1, \dots, M\}$, and its average error probability is
\begin{equation}\label{eq:avg-error}
  \varepsilon(f, g) := \frac{1}{M} \sum_{m=1}^{M}
  \sum_{y^n :\, g(y^n) \neq m} W^{\otimes n}\big( y^n \mid f(m) \big).
\end{equation}
We write $\varepsilon^*(n, M, W)$ for the minimum of \eqref{eq:avg-error} over all $(n, M)$ codes. The \emph{reliability function} of $W$ at rate $R > 0$ is the optimal exponential decay rate of the error probability of codes of rate $R$,
\begin{equation}\label{eq:reliability}
  E(R, W) := \limsup_{n \to \infty}\; -\frac{1}{n} \log
  \varepsilon^*\big( n, \lceil e^{nR} \rceil, W \big),
\end{equation}
with the convention $-\log 0 = +\infty$, so that $E(R, W) = \infty$ on rates where zero-error communication is possible. We use $\limsup$ because the limit is not known to exist at all rates. Expurgation gives a code retaining at least half the messages whose maximal error probability is at most twice the original average error probability; these constant factors do not affect the asymptotic reliability function~\cite{Gal68}. In this notation, the classical bounds read
\begin{equation}\label{eq:exp-meaning}
  E_{\mathrm{r}}(R, W) \;\le\; E(R, W) \;\le\; E_{\mathrm{sp}}(R, W),
  \qquad 0 < R < C(W),
\end{equation}
by~\cite{gallager1965simple} and~\cite{shannon1967lower}, and the three terms are equal for $R$ above the critical rate and below capacity. For $R > C(W)$, Arimoto's bound gives an upper bound $\exp(-nE_{\mathrm{sc}}(R,W))$ on the success probability at every block length, and the exponent is asymptotically attained by optimal codes~\cite{arimoto1973converse,DueckKorner1979}.

\begin{corollary}\label{cor:curves}
If $W'$ is non-signaling simulable from $W$, then for every rate $R \ge 0$,
\begin{equation}\label{eq:cor-curves}
  E_{\mathrm{r}}(R, W') \le E_{\mathrm{r}}(R, W), \qquad
  E_{\mathrm{sp}}(R, W') \le E_{\mathrm{sp}}(R, W), \qquad
  E_{\mathrm{sc}}(R, W') \ge E_{\mathrm{sc}}(R, W).
\end{equation}
\end{corollary}

\begin{proof}
In \eqref{eq:Er} and \eqref{eq:Esp} every term inside the supremum is nondecreasing in $R_\alpha$; in \eqref{eq:Esc} every term is nonincreasing in $R_\alpha$. Theorem~\ref{thm:main} gives $R_\alpha(W') \le R_\alpha(W)$ for each $\alpha$, and taking suprema term by term gives \eqref{eq:cor-curves}.
\end{proof}

Corollary~\ref{cor:curves} shows that non-signaling simulation cannot increase the random-coding or sphere-packing exponent, or decrease the strong-converse exponent. These comparisons hold separately at every rate. For rates above the critical rate of the base channel $W$ and below $C(W)$, the reliability function itself is monotone, as explained in Section~\ref{sec:disc-open}. The corresponding finite-blocklength bounds have consistent orderings: for fixed $n$, $R$, and $\rho\in[0,1]$, Gallager's upper bound $\exp(-n[\max_P E_0(\rho,P,W)-\rho R])$ on achievable average error probability is no smaller when evaluated on $W'$ than on $W$. Conversely, Arimoto's upper bound $\exp(-nE_{\mathrm{sc}}(R,W))$ on success probability is no larger when evaluated on $W'$ than on $W$. These are comparisons of bounds, not of the exact optimal finite-blocklength error probabilities.

Below the critical rate, the reliability function is not known in general. In this regime, Corollary~\ref{cor:curves} establishes monotonicity of the two bounds in \eqref{eq:exp-meaning}, but does not establish monotonicity of the reliability function itself. We further discuss this regime in Section~\ref{sec:disc-open}.

\subsection{A strong converse for non-signaling assisted codes}
\label{sec:disc-sc}

In this subsection, we derive finite-blocklength bounds on the success probability of non-signaling assisted codes from R\'enyi-capacity monotonicity. The resulting equality of the assisted and unassisted strong converse exponents is already covered by the channel-interconversion results of Oufkir, Yao and Berta~\cite[Theorem~3 and Remark~1]{Oufkir2026Strong}. Here we recover this conclusion through the capacity comparison.

An assisted code is itself a simulation. A non-signaling assisted $(n, M, \bar\varepsilon)$ code for $W$ is a non-signaling simulation, in the sense of Definition~\ref{def:ns-simulation}, of a channel $W'$ with $\mathcal{X}' = \mathcal{Y}' = \{1, \dots, M\}$ from the base channel $W^{\otimes n}$, where the average success probability is
\begin{equation}\label{eq:code-success}
  1 - \bar\varepsilon(W') = \frac{1}{M} \sum_{m=1}^{M} W'(m \mid m).
\end{equation}
The no-backward-signaling condition \eqref{eq:ns1} states that the encoder does not see the channel output, and the no-forward-signaling condition \eqref{eq:ns2} states that the decoder does not see the message. Unassisted codes, codes with shared randomness, and entanglement-assisted codes are special cases.

We use two facts about the R\'enyi capacity. First, it is additive on product channels~\cite{csiszar1995,nakiboglu2019renyi}:
\begin{equation}\label{eq:additive}
  R_\alpha(W^{\otimes n}) = n \, R_\alpha(W).
\end{equation}
Second, the radius \eqref{eq:info-radius} equals the maximum of Sibson's mutual information over input distributions~\cite{csiszar1995}:
\begin{equation}\label{eq:sibson}
  R_\alpha(W) = \max_{P \in \mathcal{P}(\mathcal{X})} I_\alpha(X, Y)_{WP},
  \qquad
  I_\alpha(X, Y)_{WP} := \frac{\alpha}{\alpha-1} \log \sum_{y}
  \Big( \sum_{x} P(x) \, W(y \mid x)^{\alpha} \Big)^{1/\alpha}.
\end{equation}
With these, we can show the following:

\begin{corollary}\label{prop:sc}
The average success probability of every non-signaling assisted $(n, M, \bar\varepsilon)$ code $W'$ for $W$ satisfies, for every $\alpha > 1$,
\begin{equation}\label{eq:prop-sc}
  1 - \bar\varepsilon(W') \le
  \exp\Big( -\frac{\alpha-1}{\alpha}
  \big( \log M - n R_\alpha(W) \big) \Big).
\end{equation}
Consequently, with $R = \frac{1}{n} \log M$,
\begin{equation}\label{eq:prop-sc-exp}
  1 - \bar\varepsilon(W') \le \exp\big( -n \, E_{\mathrm{sc}}(R, W) \big).
\end{equation}
\end{corollary}

\begin{proof}
Let $U_M$ be the uniform distribution on $\{1,\dots,M\}$ and fix $\alpha > 1$, so that $\alpha/(\alpha-1) > 0$. By \eqref{eq:sibson},
\begin{align}
  R_\alpha(W')
  &\ge I_\alpha(U_M, W')
   = \frac{\alpha}{\alpha-1} \log \sum_{y'=1}^{M}
     \Big( \frac{1}{M} \sum_{x'=1}^{M} W'(y' \mid x')^{\alpha}
     \Big)^{1/\alpha} \label{eq:sc1} \\
  &\ge \frac{\alpha}{\alpha-1} \log \sum_{m=1}^{M}
     \Big( \frac{1}{M} \, W'(m \mid m)^{\alpha} \Big)^{1/\alpha}
     \label{eq:sc2} \\
  &= \frac{\alpha}{\alpha-1} \log \Big( M^{-1/\alpha}
     \sum_{m=1}^{M} W'(m \mid m) \Big) \label{eq:sc3} \\
  &= \frac{\alpha}{\alpha-1} \log \big( M^{1-1/\alpha}
     (1 - \bar\varepsilon(W')) \big)
   = \log M + \frac{\alpha}{\alpha-1} \log(1 - \bar\varepsilon(W')).
     \label{eq:sc4}
\end{align}
In \eqref{eq:sc2} we kept only the term $x' = y'$ in the inner sum, which is allowed because all terms are nonnegative, the logarithm is increasing, and the prefactor is positive. In \eqref{eq:sc4} we used \eqref{eq:code-success}. Lemma~\ref{lem:alphagt1} applied to the simulation of $W'$ from $W^{\otimes n}$, together with \eqref{eq:additive}, gives $R_\alpha(W') \le n R_\alpha(W)$. Combining this with \eqref{eq:sc4},
\begin{equation}\label{eq:sc5}
  \log(1 - \bar\varepsilon(W')) \le \frac{\alpha-1}{\alpha}
  \big( n R_\alpha(W) - \log M \big),
\end{equation}
which is \eqref{eq:prop-sc}. Taking the supremum over $\alpha > 1$ with $\log M = nR$ gives \eqref{eq:prop-sc-exp} by \eqref{eq:Esc}.
\end{proof}

Corollary~\ref{prop:sc} is the precise meaning of ``non-signaling assistance cannot weaken the strong converse''. The bounds \eqref{eq:prop-sc} and \eqref{eq:prop-sc-exp} hold at every block length $n$; the tightness of the optimized exponent in \eqref{eq:prop-sc-exp} is an asymptotic statement. They are informative for rates $R$ above capacity: for $R \le C(W)$ every $R_\alpha(W)$ with $\alpha > 1$ is at least $C(W)$, so $E_{\mathrm{sc}}(R, W) = 0$ and \eqref{eq:prop-sc-exp} is trivial, while for every $R > C(W)$ it forces the success probability of every assisted code to decay exponentially. Arimoto's exponent $E_{\mathrm{sc}}(R, W)$ for $R>C(W)$ is asymptotically attained by unassisted codes~\cite{DueckKorner1979}. Corollary~\ref{prop:sc} shows that non-signaling assistance cannot decrease this exponent, recovering the equality of the strong converse exponents with and without non-signaling assistance~\cite[Theorem~3 and Remark~1]{Oufkir2026Strong}. As particular cases of non-signaling assistance, the same conclusion holds for shared classical randomness and entanglement.

Finally, a closed-form statement for the case where $\alpha=\infty$ can be derived. Both Corollary~\ref{prop:sc} and Corollary~\ref{cor:alphainf-code} bound the average success probability of the non-signaling simulated channel $W'$ in terms of quantities associated with the original channel $W$.

\begin{corollary}\label{cor:alphainf-code}
The average success probability of every non-signaling assisted $(n, M, \bar\varepsilon)$ code $W'$ for $W$ satisfies
\begin{equation}\label{eq:sc-inf}
  1 - \bar\varepsilon(W') \le\frac{1}{M}
  \Big( \sum_{y} \max_{x} W(y \mid x) \Big)^{n}.
\end{equation}
\end{corollary}

\begin{proof}
Let $W'$ be the channel on $\{1, \dots, M\}$ that the code simulates from $W^{\otimes n}$. By the closed form \eqref{eq:rinf-closed-form}, $R_\infty(W')$ is the logarithm of the sum of the column maxima of $W'$. Each column maximum dominates the diagonal entry of its column, so by \eqref{eq:code-success},
\begin{equation}\label{eq:inf-lower}
  R_\infty(W')
  = \log \sum_{y'=1}^{M} \max_{x'} W'(y' \mid x')
  \ge \log \sum_{m=1}^{M} W'(m \mid m)
  = \log \big( M (1 - \bar\varepsilon(W')) \big).
\end{equation}
In the other direction, Lemma~\ref{lem:alphainf} applied to the simulation of $W'$ from the base channel $W^{\otimes n}$ gives $R_\infty(W') \le R_\infty(W^{\otimes n})$, and $R_\infty$ is exactly additive on products, because the maximum over product inputs is attained coordinate by coordinate:
\begin{equation}\label{eq:inf-tensor}
  \sum_{y^n} \max_{x^n} W^{\otimes n}(y^n \mid x^n)
  = \sum_{y^n} \prod_{i=1}^{n} \max_{x_i} W(y_i \mid x_i)
  = \Big( \sum_{y} \max_{x} W(y \mid x) \Big)^{n}.
\end{equation}
Combining \eqref{eq:inf-lower} with $R_\infty(W') \le n R_\infty(W)$ and exponentiating gives \eqref{eq:sc-inf}.
\end{proof}
\section{Discussion}
\label{sec:discussion}

\subsection{What non-signaling assistance can still do}
\label{sec:disc-coding}

The phrase ``cannot enhance reliability'' must be read with care. Theorem~\ref{thm:main} is about the simulation of channels. Assisted coding is a different task, and assistance can improve the error exponent below capacity. The results below show that non-signaling assisted codes attain the sphere-packing exponent at every positive rate below capacity, whereas unassisted codes need not attain it at low rates. Corollary~\ref{cor:curves} states that this sphere-packing exponent cannot increase under non-signaling simulation of the channel.

Matthews~\cite{matthews2012linear} showed that the optimal error probability of a non-signaling assisted code for a classical channel equals the meta-converse bound of~\cite{polyanskiy2010channel}. Polyanskiy~\cite[Th.~24]{polyanskiy13} showed that the exponent of this bound is the sphere-packing exponent, so the non-signaling assisted reliability function $E^{\mathrm{NS}}(R, W)$ of a classical channel is
\begin{equation}\label{eq:ns-reliability}
  E^{\mathrm{NS}}(R, W) = E_{\mathrm{sp}}(R, W)
  \qquad \text{for all } 0 < R < C(W),
\end{equation}
with no critical rate, and with zero error at rates below $R_0(W)$ (Subsection~\ref{sec:disc-zero}). Unassisted codes in general do not reach \eqref{eq:ns-reliability} at low rates. For the binary symmetric channel $W_p$ with crossover probability $p \in (0, \tfrac12)$, for example,
\begin{equation}\label{eq:bsc}
  \lim_{R \downarrow 0} E^{\mathrm{NS}}(R, W_p)
  = -\log\big( 2\sqrt{p(1-p)} \big),
  \qquad
  \lim_{R \downarrow 0} E(R, W_p)
  = -\tfrac12 \log\big( 2\sqrt{p(1-p)} \big),
\end{equation}
where $E$ is the unassisted reliability function and the second limit is the zero-rate value obtained from the matching classical bounds~\cite{shannon1967lower,Gal68}. Non-signaling assistance therefore doubles the zero-rate exponent of the binary symmetric channel, while it leaves the capacity unchanged, as does entanglement~\cite{bennett2002}.

For classical-quantum channels, Oufkir, Tomamichel and Berta~\cite{oufkir2024} established the corresponding sphere-packing characterization for activated non-signaling assistance, which includes an additional noiseless bit of communication. Their results without activation apply under additional channel assumptions or rate restrictions. This distinction is necessary when comparing the classical and classical-quantum settings.

\subsection{Open questions}
\label{sec:disc-open}

\emph{Monotonicity of the reliability function at low rates.} Below the critical rate, the reliability function is not known in general for unassisted codes: the random-coding lower bound \eqref{eq:Er} and the sphere-packing upper bound \eqref{eq:Esp} separate~\cite{gallager1965simple,shannon1967lower}.

The relevant rates satisfy
\begin{equation}\label{eq:rate-chain}
  C_0(W) \le R_0(W) \le R_{\mathrm{crit}}(W) \le C(W),
\end{equation}
where the critical rate $R_{\mathrm{crit}}(W)$ is the smallest rate at which \eqref{eq:Er} and \eqref{eq:Esp} agree. The middle inequality holds because $E_{\mathrm{sp}}(\cdot\,, W)$ is infinite below $R_0(W)$ while \eqref{eq:Er} is finite. We can consider monotonicity of the reliability function separately on the resulting intervals. Below $C_0(W)$, $E(R,W)=\infty$, so the inequality $E(R,W')\le E(R,W)$ is trivial. Above $R_{\mathrm{crit}}(W)$, monotonicity of the reliability function follows from Corollary~\ref{cor:curves}: for $R \in (R_{\mathrm{crit}}(W), C(W))$,
\begin{equation}\label{eq:high-rate-mono}
  E(R, W') \le E_{\mathrm{sp}}(R, W') \le E_{\mathrm{sp}}(R, W)
  = E(R, W).
\end{equation}

For $R\in(C_0(W),R_0(W))$, zero-error assisted coding gives counterexamples to monotonicity. Fix $R\in(C_0(W),R_0(W))$. Non-signaling boxes allow zero-error communication over $W^{\otimes n}$ at rate $R$ for sufficiently large $n$~\cite{cubitt2011zero}, since the non-signaling-assisted zero-error capacity equals $R_0(W)$. Such a code is itself a non-signaling simulation of a noiseless channel $W'$ with $\lceil e^{nR}\rceil$ symbols from $W^{\otimes n}$. For every $R'\in(C_0(W),R)$, $E(nR',W')=\infty$, whereas $E(nR',W^{\otimes n})=nE(R',W)$ is finite. This does not contradict Corollary~\ref{cor:curves}, because $E_{\mathrm{sp}}(R',W)$ is itself infinite there, so the ordering of the upper bounds is trivial. This interval can be nonempty~\cite{cubitt2011zero}.

Finally, in the remaining interval $(R_0(W),R_{\mathrm{crit}}(W))$, both bounds are finite and ordered, but their separation prevents us from deducing monotonicity of the reliability function from Corollary~\ref{cor:curves}. To the best of our knowledge, monotonicity in this regime remains open.

\emph{Classical-quantum channels.} Oufkir, Yao and Berta~\cite[Theorem~12]{Oufkir2026Strong} already determined strong converse exponents for classical-quantum channel interconversion under entanglement and non-signaling assistance, using the purified distance. A possible extension of the present approach is to adapt our variational proof to channels $x\mapsto\rho_x$ under non-signaling quantum superchannels that map classical-quantum channels to classical-quantum channels. Such an extension should distinguish capacities defined through the Petz and sandwiched R\'enyi divergences and specify the range of orders under consideration. Related operator variational representations based on Young's inequality already appear in~\cite{oufkir2024}. Adapting these methods to the capacity comparison requires treating the noncommuting output operators and the quantum non-signaling constraints together. We leave this extension of the proof method, and a precise comparison with the existing channel-interconversion results, to future work.

\section*{Acknowledgments}
T.S. acknowledges support from the Centre for Quantum Technologies at the National University of Singapore through the Alice Prize 2025. M.H.R. is supported by the National Research Foundation, Prime Minister's Office, Singapore, under its Campus for Research Excellence and Technological Enterprise (CREATE) program, Quantum Security and Resilience for Emerging Technologies Enterprise (QUASAR-CREATE). M.T. acknowledges support from the National Research Foundation Investigatorship Award (NRF-NRFI10-2024-0006) and the National Research Foundation, Singapore, through the National Quantum Office, hosted by A*STAR, under its Centre for Quantum Technologies Funding Initiative (S24Q2d0009).

\section*{LLM use disclaimer}
Claude assistance has been used to derive the proofs of Lemmas 2 to 5 and to assist in writing the manuscript. The authors claim full responsibility over the statements in the manuscript, including any mistakes present.

\printbibliography
\end{document}